\documentclass[a4paper,twoside,11pt,final]{article}
\usepackage[utf8]{inputenc}
\usepackage{amssymb, amsthm}
\usepackage{amsmath}
\usepackage{graphicx}
\usepackage{cite}
\newtheorem{theorem}{Theorem}
\newtheorem{lemma}{Lemma}

\newtheorem{remark}{Remark}
\usepackage[toc,page]{appendix}
\usepackage{times}
\usepackage{caption}
\theoremstyle{remark}
\begin{document}
\title{A Probabilistic Baby-Step Giant-Step Algorithm} 
\author{Prabhat Kushwaha\\
prabhatkk@students.iiserpune.ac.in\\
Ayan Mahalanobis\\
ayan.mahalanobis@gmail.com\\
\small{IISER Pune, Dr.~Homi Bhabha Road, Pune 411008, INDIA}}
\date{}
\bibliographystyle{plain}
\maketitle
\begin{abstract}
In this paper, a new algorithm to solve the discrete logarithm problem is presented which is similar to the usual baby-step 
giant-step algorithm.  
Our algorithm exploits the order of the discrete logarithm in the multiplicative group of a finite field. 
Using randomization with parallelized collision search, our algorithm indicates some weakness in NIST 
curves over prime fields which are considered to be the most conservative and safest curves among all NIST curves. 
\end{abstract}

 
\section{Introduction}
It is well-known that computationally hard number theoretic problems are used as primitives in public-key cryptography. 
On that basis, public-key cryptography can be divided into two categories. One uses the hardness of factorizing large
integer as the building blocks to construct public-key protocols and the other is based on the computational difficulty
of solving the discrete logarithm problem. In this paper, we are interested in the latter.

Let $G$ be a cyclic group of prime order $p$ and generated by $P$ which is written additive. Given an element $Q=xP \in
\mathbb{G}$, the \emph{discrete logarithm problem}(\textbf{DLP}) in $G$ is to compute the integer $x$. This integer $x$ 
is called the discrete logarithm of $Q$ with the base $P$. There are generic algorithms such as the baby-step
giant-step algorithm~\cite{sil} which solves DLP in any group $G$. 

In this paper, we develop and study a different version of the baby-step giant-step algorithm. The novelty of our approach 
comes from the \emph{implicit representation} using $F_p^\times$ as \emph{auxiliary group}. Our approach leads to a way to
reduce the discrete logarithm problem to a problem in $\mathbb{F}_p^\times$. The advantage of this approach is, 
$\mathbb{F}_p^\times$ has many subgroups and one can exploit the rich and well understood subgroup structure of
$\mathbb{F}_p^\times$. 

In Theorem 1 we develop an algorithm that solves the discrete logarithm problem using implicit representation. Two things 
come out of this theorem:
\begin{description}
\item[A] If the secret key $x$ belongs to some small subgroup of $\mathbb{F}_p^\times$, there can be an efficient attack 
on the DLP.
\item[B] If somehow it is known to an attacker that the secret key is in some subgroup $H$ of $\mathbb{F}_p^\times$, that
information can be used to develop a better attack. 
\end{description}

The question remains, what happens if no information about the secret $x$ is known. We develop a probabilistic algorithm 
(Theorem 2) to expand our attack. To understand this probabilistic attack properly, we study it on the curve P-256. This
is an NIST recommended curve over a prime field and is considered secure. 
Our study, which we present in details in Section 3 indicates some weakness in this curve.

\section{Main Work}
Let $G$ be a cyclic group of prime order $p$ and generated by $P$ which is written additive. For $y\in\mathbb{F}_p$, 
$yP \in G$ is called the implicit representation of $y\in \mathbb{F}_p$(with respect to $G$ and $P$).
The following lemma comes from the idea of implicit representation of a finite field, proposed by Maurer and Wolf
~\cite{maurer1999relationship}. 
\begin{lemma}
Let $a,b$ be any two integers. Then $a=b\pmod p$ if and only if $aP=bP$ in $G$.
\end{lemma}
\begin{proof}
Assume that $a=b\pmod p$, then $a=tp+b$ for some integer $t$. Then $aP=tpP+bP=bP$.
Conversely, assume that $aP=bP$, then $(a-b)P=0$ in $G$ and this means $p|(a-b)$ which implies that $a=b\pmod p$.
\end{proof}

The usefulness of this lemma is to be able to decide on the equality in $\mathbb{F}_p^\times$ by looking at the equality in
$G$. The following algorithm to solve the discrete logarithm problem uses the order of the discrete logarithm in the multiplicative
group of a finite field. This algorithm is different from the baby-step giant-step~\cite{sil} as it uses the implicit
representation with multiplicative group of a finite field as auxiliary group. 

\begin{theorem}
Let $G$ be an additive cyclic group generated by $P$ and order of $P$ is a prime $p$. Let $Q = xP$ be another given 
element of $G$($x$ is unknown). For a given divisor $d$ of $p-1$, let $H$ be the unique subgroup of 
$\mathbb{F}_p^\times$ of order $d$. Then, one can decide whether or not $x$ belongs to $H$ in $O(\sqrt{d})$ steps.
Furthermore, if $x$ belongs to $H$, the same algorithm will also find the discrete logarithm $x$ in 
$O(\sqrt{d})$ steps where each step is an exponentiation in the group $G$.
\end{theorem}
\begin{proof}
Since $H$ is a subgroup of the cyclic group $\mathbb{F}_p^\times$, we assume that it is generated by some element $\zeta$. If 
the generator of $H$ is not given to us, we can compute it using a generator of $\mathbb{F}^{\times}$ and $d$. 
The proof of whether $x$ belongs to $H$ or not follows from the well-known baby-step giant-step algorithm
~\cite[Proposition 2.22]{sil} to compute the discrete logarithm.

Let $n$ be the smallest integer greater than $\sqrt{d}$. Then $x\in \mathbb{}H$ if and only if there exists an integer $k$ with
$0 \leq k \leq d$ such that $x=\zeta^k \pmod p$. Note that any integer $k$ between $0$ and $d$ can be written as $k = an -b$ for
unique integers $a, b$ with $0 \leq a, b \leq n$, by division algorithm. Therefore,  $x\in \mathbb{}H$ if and only if there exist
two integers $a, b$ with $0\leq a,b \leq n$ such that $x=\zeta^{an -b}\pmod p$, or equivalently $\zeta^b x={\zeta^n}^a \pmod p$.
Using the lemma above, we see that  $x\in \mathbb{}H$ if and only if there exist two integers $a, b$ with $0\leq a,b \leq n$ 
such that $\zeta^b xP={\zeta^n}^a P$, equivalently  $\zeta^b Q=(\zeta^n)^a P$ as $Q = xP$.  

Now, we create a list $\left\{\zeta^bQ:\;0\leq b\leq n\right \}$. Then we generate elements of the form ${(\zeta^n)^a}P$ for 
each integer $a$ in $[0,n]$ and try to find a collision with the earlier list. When there is a collision, i.e., 
$\zeta^b Q={(\zeta^n)^a}P$ for some $0\leq a,b \leq n$, it means that $x \in H$. Otherwise, $x \notin H$.

Moreover, if $x \in H$ then $\zeta^b Q={(\zeta^n)^a}P$ for some $0\leq a,b \leq n$. So, we use the integers $a$ and 
$b$ to compute $\zeta^{an -b} \pmod p$ which is nothing but the discrete logarithm $x$. 
Since the two lists require computation of at most $2n$ exponentiations, the worst case time complexity of the algorithm to 
check whether or not $x \in H$, as well as to compute $x$(if $x\in H$) would be $O(n) \approx 
O(\sqrt{d})$ steps. This completes the proof.
\end{proof}

\begin{remark}
Even though the above algorithm is generic in nature, it does have a practical significance. Our algorithm applies on
all the five prime order NIST curves~\cite{nist} viz. P-192, P-224, P-256, P-384, P-521. Although the probability of a randomly
chosen secret key $x$ being inside a particular subgroup of $\mathbb{F}_p^\times$ can be very small, however, it is advisable
to check, using our algorithm for each curve, if the secret key $x$ belongs to any of two (large enough)subgroups whose orders
are mentioned in the appendix A. If it does, we discard the secret key.   
\end{remark}
Suppose that $p-1$ has large enough(but a lot smaller than $p-1$) divisor $d$ and $H$ is the unique subgroup of 
$\mathbb{F}_p^\times$ of order $d$. A drawback of the deterministic algorithm given in Theorem 1 is that it might fail to solve 
DLP because the probability of $x$ belonging to $H$ is very small. One way to increase the probability is to increase 
the size of $d$, if such $d$ exists. Clearly, this is not a desirable solution because the computational cost depends on the size 
of the subgroup. 

The above algorithm can be parallelized which helps us overcome this obstacle by increasing 
the probability. We have \textit{randomized} the above algorithm where the random inputs will be 
running on parallel processes or threads. This parallelization along with collision algorithm (based on birthday paradox)~\cite
[Theorem 5.38]{sil} yields a randomized probabilistic algorithm which can solve DLP with a given probability. 

\noindent\textbf{Collision Theorem}: An urn contains $N$ balls, of which $n$ balls are red and $N-n$ are blue. One randomly 
selects a ball from the urn, replaces it in the urn, randomly selects a second ball, replaces it, and so on. He does this until
he has looked at a total number of $m$ balls. Then, the probability that he selects at least one red ball is
 $$\text{Pr}(at \; least \; one \;  red\; ball) = 1 - \left(1-\frac{n}{N}\right)^m \geq 1- e^{\frac{-mn}{N}}.$$

\begin{theorem}
Let $G$ be an additive cyclic group generated by $P$ and the order of $P$ is a prime $p$. Let $Q = xP$ be another 
given element of $G$($x$ is unknown). For a given divisor $d$ of $p-1$, let $H$ be the unique subgroup of 
$\mathbb{F}_p^\times$ of order $d$. Then, $x$ can be computed in $O(\sqrt{d})$ steps with probability at least 
$1-e^{\left(\frac{-dm}{p-1}\right)}$ if one has access to $m$ parallel threads.
\end{theorem}
\begin{proof}
  The main idea is to run the algorithm in Theorem 1 on each of $m$ threads as follows. We randomly selects $m$ elements 
$y_1, y_2, .., y_m$ in $\mathbb{F}_p^\times$ and compute corresponding $m$ elements $Q_1=y_1Q=(y_1x)P$,...,$Q_m=y_mQ=(y_mx)P$ 
of $G$. Now, we run the above algorithm on each of $m$ parallel threads, with element $Q_i=(y_ix)P$ 
running on $i^{th}$ thread. Let $z_i= y_i x \pmod p$ for $i =1,..,m$. If $z_i\in H$ for some $i$, $1\leq i \leq m$; 
then the algorithm on that thread returns $z_i$. Once we have $z_i$ for some $i$, we compute $z_i\cdot {y_i}^{-1} \pmod p$ 
which is nothing but the discrete logarithm $x$.

The collision theorem above tells us about the probability of at least one $z_i$ belonging to $H$ for $1\leq i\leq m$.
In present case, $\mathbb{F}_p^\times$ with $p-1$ elements is the urn, so $N =p-1$. The elements of $H$ are red balls,
so $n=d$. Since we are randomly selecting $m$ elements $y_1,..,y_m$ from $\mathbb{F}_p^\times$, it implies that $z_1, z_2,..,
z_m$ also are random elements of $\mathbb{F}_p^\times$. Therefore, probability that at least one of $z_i$ would belong to 
$H$ is at least $1-e^{\left(\frac{-dm}{p-1}\right)}$, by the collision theorem. In other words, with probability at least 
$1-e^{-\frac{dm}{p-1}}$, one can compute $z_i$ for some $i$, $1\leq i \leq m$ if one has access to $m$ threads. Since the
number of steps performed on each thread before $z_i$ is computed for some $i$ is at max $2\sqrt{d}$, we conclude that
it takes $O(\sqrt{d})$ steps to compute $x$ with the probability at least $1-e^{\left(\frac{-dm}{p-1}\right)}$ if $m$ threads
are available. This completes the proof.
 \end{proof}
\begin{remark}
It follows from Theorem 2 that if there exist divisors $d$ of $p-1$ of suitable sizes, then DLP can be solved in time much less
than the square root of the group size but with a probability which increases with the number of threads used. A practical
importance of Theorem 2 lies in the fact that such divisors of $p-1$ do exist for all NIST curves~\cite{nist} as well as most of
SEC2 curves~\cite{secg}. This gives us precise estimates about the number of group operations and threads needed to solve DLP
with a given probability. We illustrate this by an example in the next section. 
\end{remark}
\begin{remark}
Note that the probability of solving the DLP in above theorem is proportional to the product $m \cdot d$. It follows that if
we fix a probability, this product is constant. Therefore, for a fixed probability of solving the DLP,
there is a trade-off between the number of steps and number of threads needed in Theorem 2. Increasing one of the two would
decrease the other and vice-a-versa. 
\end{remark}

\section{Security analysis of NIST curve P-256}
As discussed earlier, our probabilistic algorithm is applicable to NIST curves. In this section, we will demonstrate 
the implication of our algorithm on NIST curves. We will do that only on the NIST curve P-256 but similar conclusions hold 
for other four NIST curves over prime field as well, see appendix.\\
The NIST curve P-256 is defined over the prime field $\mathbb{F}_q$ and the order of P-256 is a prime $p$ given below.\\
\newline

$q=1157920892103562487626974469494075735300861434152903141955$\\
$33631308867097853951$\\
$p=115792089210356248762697446949407573529996955224135760342422$\\
$259061068512044369 $\\
$p-1 =2^4\cdot3\cdot71\cdot131\cdot373\cdot3407\cdot17449\cdot 38189 \cdot 187019741\cdot622491383\cdot$\\
$1002328039319\cdot 2624747550333869278416773953$\\
\newline
Since $p-1$ factors into many relatively small integers, we have the following divisors of $p-1$ of various sizes.\\

$d_1=534427449503294145963994143640970973102047412378826412971$\\
$9829 \approx 2^{201.73}$. \\

$d_2=106885489900658829192798828728194194620409482475765282594$\\
$39658\approx 2^{202.73}$. \\

$d_3=160328234850988243789198243092291291930614223713647923891$\\
$59487\approx 2^{203.32}$. \\

$d_4 =18207943204577231552993280473847881053586755339746615$\\$
889955457403\approx2^{213.47}$.\\

$d_5=238524055979961733344211974207407241801986494950680668158$\\
$4164919793\approx2^{220.50}$

For above sizes of subgroups and various number of threads $m$, the following tables give the probability to solve DLP. The 
second column of the Table 1 shows the probabilities when the subgroup size is $d_1 \approx 2^{201.73}$ bits. For example, if 
we have $m= 2^{54}$ parallel threads, then our algorithm would solve DLP in $2^{101.86}$ steps with probability $0.56458$ which
is the intersection of the fifth row(corresponding to $m=2^{54}$) and the second column(corresponding to $d_1 \approx 
2^{201.73}$). Other entries(probabilities) of the tables can be understood similarly. 

If we go across a row in the tables, we see the probabilities getting increased with the size of subgroup $d$. If we move 
along a column, probabilities increase with the number ($m$) of parallel threads. Table 1 also exhibits the trade-off 
between $d$ and $m$ for equal probability. For equal probability, highlighted diagonally in
the second and third column, we see that increasing the subgroup size by $1$-bit($d_1$ and $d_2$ differ by $1$-bit)
results in a decrease of $1$-bit in the number of parallel threads $m$. As an 
example, to achieve the probability $0.56458$, the subgroup of order $d_1$ requires $2^{54}$ parallel threads while 
the subgroup of order $d_2$ requires $2^{53}$.
\begin{center}
\begin{table}
\caption{}
\begin{tabular}{ |c| c| c|c|c|}
 \hline                  & $\log_2{d_1}=201.73$ & $\log_2 {d_2}=202.73$ & $\log_2 {d_3} = 203.32$ \\
                         & $\log_2 (\sqrt{d_1})=101.86$   & $\log_2(\sqrt{d_2})=101.36$    & $\log_2(\sqrt{d_3})=101.66$\\
 \hline  $\log_2 m =45$   & $0.00162$   & $0.00324$      &  $0.00486$ \\ 
 \hline  $\log_2 m =50$   &$0.05064$    & $0.098711$     & $0.14435$  \\ 
 \hline  $\log_2 m =52$   &$0.18768$    &  \textbf{0.34013}     &    $0.46398$\\ 
 \hline  $log_2 m =53$   &\textbf{0.34013}    & \textbf{0.56458}       &  $0.71268$ \\ 
 \hline  $\log_2 m =54$   &\textbf{0.56458}    & \textbf{0.81040}      &  $0.91745$ \\ 
 \hline  $log_2 m =55$   &\textbf{0.81040}    & \textbf{0.96405}     & $0.993184$ \\ 
 \hline  $log_2 m =56$   & \textbf{0.96405}   & $0.99871$     &$0.99995$ \\ 
 \hline
\end{tabular}
\end{table}
\end{center}
\begin{table}
\centering
\begin{minipage}{0.48\textwidth}
\centering
\begin{tabular}{ |c| c|}
 \hline    & log$_2 {d_4}=213.47$ \\
           & log$_2 (\sqrt{d_4})=106.78$\\
\hline log$_2 m =41$ & $0.29234$\\
\hline log$_2 m =42$ & $0.49921$ \\
\hline log$_2 m =43$ & $0.74921$ \\
\hline log$_2 m =44$ & $0.93710$\\
\hline
\end{tabular}
\caption{}
\end{minipage}
\hfill
\begin{minipage}{0.48\textwidth}
\centering
\begin{tabular}{ |c| c|}
 \hline    & log$_2 {d_5}=220.50$ \\
           & log$_2 (\sqrt{d_5})=110.25$\\
\hline log$_2 m =33$ & $0.16218 $\\
\hline log$_2 m =34$ & $0.29805$ \\
\hline log$_2 m =35$ & $0.50727$ \\
\hline log$_2 m =36$ & $0.75721$\\
\hline log$_2 m =37$ & $0.94106$\\
\hline
\end{tabular}
\caption{}
\end{minipage}
\end{table}

From Table 3, we can see that DLP on the curve P-256 can be solved in $2^{110.25}$(with a significant reduction from 
$2^{128}$) steps with probability greater than $0.5$, while using $2^{35}$ parallel threads. This indicates a weakness
of NIST curve P-256 if one assumes that $2^{35}$ parallel threads are within the reach of modern distributed 
computing. Similar conclusions can be drawn for other NIST curves P-192, P-224, P-384 and P-521 see appendix.

Moreover, one observes that for most of the curves in SEC2(Version 2)~\cite{secg} which also include all other ten NIST 
curves~\cite{nist}over binary field, $p-1$ factors into small divisors. Therefore, 
our algorithm for solving DLP on those curves in SEC2~\cite{secg} can similarly be studied.
\section{Conclusion}
In this paper we presented a novel idea of using the implicit representation with $\mathbb{F}_p^\times$ as
auxiliary group to solve the discrete logarithm problem in a group $\mathbb{G}$ of prime 
order $p$. We modified the most common generic algorithm, the baby-step giant-step algorithm for this purpose and studied it
further for NIST curves over prime fields. This algorithm
that we developed brings to the spotlight the structure of the auxiliary group for the security of the discrete logarithm
problem in $G$. This aspect is probably reported for the first time. 
\bibliography{Nist}
\nocite{*}
\begin{appendices}
\appendix
\section{NIST Curves Over Prime Field}
For each of these five NIST curves of order prime $p$, two subgroups of $\mathbb{F}_p^\times$ with (large enough)orders $d_1$,
$d_2$ are given such that $d_1 \cdot d_2 = p-1$ and $\gcd(d_1, d_2) = 1$, see Remark 1.
\subsection{P-192}
$p= 6277101735386680763835789423176059013767194773182842284081$\\
\newline 
$p-1 = 2^4 \cdot 5 \cdot 2389 \cdot 9564682313913860059195669 \cdot 3433859179316188$\\
$682119986911 $\\
$d_1 = 656279166350909980926771898430320 \approx 2^{109.02}$\\
$d_2= 9564682313913860059195669 \approx 2^{82.98}$

\subsection{P-224}
$p= 269599466671506397946670150870196259404578077144243917216827$\\
$22368061 $\\
\newline
$p-1= 2^2 \cdot  3^6 \cdot 5 \cdot 2153 \cdot 5052060625887581870$\\
$7470860153287666700917696099933389351507 $\\
$d_1 = 50520606258875818707470860153287666700917696099933389351507 \approx 2^{195.01}$\\
$d_2= 533642580 \approx 2^{28.99}$

\subsection{P-256}
$p = 115792089210356248762697446949407573529996955224135760342422$\\
$259061068512044369 $\\
\newline 
$p-1 =2^4\cdot3\cdot71\cdot131\cdot373\cdot3407\cdot17449\cdot 38189 \cdot 187019741\cdot622491383\cdot$\\
$1002328039319\cdot 2624747550333869278416773953$\\
\newline
$d_1= 1489153224408067225170753316415649493584 \approx 2^{130.13}$\\
$d_2= 77757001302792844776776389119582520177  \approx 2^{125.87}$

\subsection{P-384}
$p = 3940200619639447921227904010014361380507973927046544666794$ \\
$6905279627659399113263569398956308152294913554433653942643$ \\
\newline
$p-1 = 2 \cdot 3^2 \cdot 7^2 \cdot 13 \cdot 1124679999981664229965379347 \cdot$\\
$3055465788140352002733946906144561090641249606160407884365391979704929$\\
$268480326390471$\\
\newline
$d_1= 1167799024227242535444914507528451248843085599474507893404452814$\\
$6432239664131807464380162 \approx 2^{292.55}$\\
$d_2= 1124679999981664229965379347 \approx 2^{89.86} $

\subsection{P-521}
$p = 686479766013060971498190079908139321726943530014330540939 4463$\\
$45918554318339765539424505774633321719753296399637136332111 386476$ \\
$8612440380340372808892707005449$\\
\newline
$p-1 = 2^3 \cdot 7 \cdot 11 \cdot 1283 \cdot 1458105463 \cdot 1647781915921980690468599\cdot$\\
$3615194794881930010216942559103847593050265703173292383701371712350878926821$\\
$661243755933835426896058418509759880171943 $\\
\newline
$d_1 = 4166083869350854498586791068944823620942931357552596820305098954973$\\
$694271292315253349654329419600683157636543108630210814256821981752 \approx 2^{440.55}$\\
$d_2= 1647781915921980690468599 \approx 2^{80.45} $
\end{appendices}
\end{document}